\documentclass[reqno,a4paper]{amsart}
\usepackage{amssymb,latexsym,amsmath,enumerate}
\usepackage[mathscr]{eucal}
\usepackage{color}
\usepackage[backref,bookmarks=true]{hyperref}
\numberwithin{equation}{section} 
\newtheorem{theorem}{Theorem}
\numberwithin{theorem}{section}
\newtheorem{lemma}{Lemma}
\numberwithin{lemma}{section}
\newtheorem{proposition}{Proposition}
\numberwithin{proposition}{section}
\newtheorem*{corollary}{Corollary}
\theoremstyle{definition}
\newtheorem{definition}{Definition}
\numberwithin{definition}{section}

\theoremstyle{remark}

\newcommand{\rd}{\mathbf{R}^d}
\newcommand{\rr}{\mathbf{R}}

\newcommand{\qp}{\mathbf{Q}_p}

\newcommand{\one}{\mathbf 1}

\newcommand{\fo}{{\mathcal F}}
\newcommand{\foi}{{\mathcal F}^{-1}}
\newcommand{\gl}{\lambda}
\newcommand{\bn}{B_n}
\newcommand{\bmn}{B_{-n}}
\newcommand{\qmn}{q^{-n}}

\newcommand{\dn}{\mathcal{D}_n}

\newcommand{\pna}{\ensuremath{P_n^\alpha}}
\newcommand{\beq}{\begin{equation}}
\newcommand{\eeq}{\end{equation}}
\newcommand{\beqstar}{\begin{equation*}}
\newcommand{\eeqstar}{\end{equation*}}	
\newcommand{\ptnt}{\ensuremath{(p_{t,n})_{t>0}}}
\newcommand{\psns}{\ensuremath{(p_{s,n})_{s>0}}}
\newcommand{\expta}{\ensuremath{e^{-t|\xi|^\alpha}}}

\newcommand{\ptn}{\ensuremath{p_{t,n}}}
\newcommand{\psn}{\ensuremath{p_{s,n}}}
\newcommand{\Pn}{\ensuremath{\mathbf{P}^n}}
\newcommand{\etpan}{\ensuremath{e^{-tP^\alpha_n}}}
\newcommand{\knt}{K^n_t}
\newcommand{\ethn}{e^{-tH_n}}

\DeclareMathOperator{\supp}{supp}
\DeclareMathOperator{\ave}{ave}

\begin{document}
	\title[Brownian Motion]{Brownian Motion and Finite Approximations of Quantum Systems over Local Fields}
	\author{Erik M. Bakken}\address{Department of Mathematical Sciences\\The Norwegian
		University of Science and Technology\\7491 Trondheim\\Norway}
	\email{erikmaki@math.ntnu.no}
	
	\author{Trond Digernes}
	\address{Department of Mathematical Sciences\\The Norwegian
		University of Science and Technology\\7491 Trondheim\\Norway}
	\email{digernes@math.ntnu.no}
	
	\author{David Weisbart}
	\address{Department of Mathematics\\University of California\\Riverside, CA 92521}
	\email{weisbart@math.ucr.edu}

	\thanks{ The first named author gratefully acknowledges support from ``Norges Tekniske H\o gskoles fond" and ``Forsknings- og undervisningsfondet i Trondheim" at NTNU, and from the Math Department at UCLA. The second named author had partial support from the Norwegian Research Council during parts of this research. Both the first and the second named author received support from the Norwegian University of Science and Technology (NTNU)} 
\begin{abstract}
	We give a stochastic proof of the finite approximability of a class of Schr{\"o}dinger operators over a local field, thereby completing a program of establishing in a non-Archimedean setting corresponding results and methods from the Archimedean (real) setting. A key ingredient of our proof is to show that Brownian motion over a local field can be obtained as a limit of random walks over finite grids. Also, we prove a Feynman-Kac formula for the finite systems, and show that the propagator at the finite level converges to the propagator at the infinite level.
\end{abstract}
\maketitle
\tableofcontents


\section{\large \bf  Introduction}
This article grew out of a desire to explore the utility and effectiveness of stochastic methods in a non-Archimedean setting. In a recent article two of us gave a functional analytic proof of the finite approximability of the Schr\"odinger operator over a local field \cite{BD15}. In the present article we give a stochastic proof of the same. The inspiration comes from \cite{DVV94}, where both a functional analytic and stochastic proof was given for the corresponding theorem over $\rd$. In both cases the stochastic method gave a stronger convergence result for the eigenfunctions (at the expense of a mild growth condition on the potential).

The results of \cite{DVV94} were later partially extended to a setting of locally compact abelian groups in \cite{AGK00}. However, the proofs of \cite{AGK00} used non-standard analysis. We have found it worthwhile to present proofs which do not rely on non-standard methods.

Non-Archimedean stochastics has been extensively explored by several authors. Kochubei has devoted a whole book to the subject \cite{Koc01}, and the long list of references therein testifies to an active field of research. For articles on non-Archimedean random walks specifically, see, e.g., \cite{AK94,AK99} and \cite{CZ13}. Of particular interest to us is the probability density induced by the non-Archimedean ``Laplacian" over a local field. The existence of this density was obtained independently by several authors, among them Kochubei \cite{Koc91} and Varadarajan \cite{Var97} (see \cite[Ch. 4]{Koc01} and \cite[Ch. XVI]{VVZ94} for further references). In this article we show that an analogous density can be defined at the finite level, and that the associated objects at the finite level converge to the corresponding objects at the infinite level.

Our setting is as follows: $K$ is a local field with canonical absolute value $|\cdot|$, and $H=P^\alpha+V$ is a Schr\"odinger operator, densely defined and self-adjoint on a suitable domain in $L^2(K)$. $V$ is the potential given as $(Vf)(x)=v(x)f(x)$ with $v:K\to[0,\infty)$ a continuous function such that $v(x)\to\infty$ as $|x|\to\infty$. $P=\foi Q\fo$ where $(Qf)(x)=|x|f(x)$, $\fo$ is the Fourier transform, and $\alpha$ is a positive real number. It is customary to refer to $P^\alpha$ as the (negative of) the non-Archimedean Laplacian for any $\alpha>0$, although it is only $\alpha=2$ which gives a direct analog. Our task is to construct finite models $X_n$ for $K$ and corresponding Schr\"odinger operators $H_n=P_n^\alpha+V_n$ on $L^2(X_n)$ such that the eigenvalues and eigenfunctions for $H_n$ converge to the corresponding objects for $H$ (in a manner to be made precise below).

The structure of the paper is as follows: In Section~2 we collect the facts we need about local fields and the finite models. In Section~3 we construct probability densities for the finite models and prove some basic facts about them. In Section~4 we use the results from Section~3 to construct measures of the Wiener type over the finite models and prove that both the conditioned and the unconditioned versions converge to the corresponding measures over the local field. In Section~5 we prove a theorem of the Feynman-Kac type associated with the stochastics at the finite level. In Section~6 we use our results to give a stochastic proof of the finite approximability of the Schr\"odinger operator over a local field.

\section{\large \bf  Basics about Local Fields and Finite Models}\label{localfields}
We recall here, without proofs, some quick facts about local fields and their finite models. For details see \cite[Section 2]{BD15}
\subsection{Local Fields}
By a local field we mean a non-discrete, totally disconnected, locally compact field. It comes equipped with a canonical absolute value which is induced by the Haar measure, and which we denote by $|\cdot|$.
There are two main types of local fields:\\
\emph{Characteristic zero.}
The basic example of a local field of characteristic zero is the $p$-adic field $\qp$ ($p$ a prime number). Every local field of characterisitic zero is a finite extension of $\qp$ for some $p$.\\
\emph{Positive characteristic.}
Every local field of positive characteristic $p$ is isomorphic to the field $\mathbf F_q((t))$ of Laurent series over a finite field $\mathbf F_q$, where $q=p^f$ for some positive integer $f\geq1$. 

Let $K$ be a local field with canonical absolute value $|\cdot|$. We use the following standard notation: 
\[
O = \{x \in K : |x|\leq 1\},\quad P = \{x \in K : |x| < 1\},\quad U = O\setminus P.
\]
$O$ is a compact subring of $K$, called the \emph{ring of integers}. It is a discrete valuation ring, i.e., a principal ideal domain with a unique  maximal ideal. $P$ is the unique non-zero maximal ideal of $O$, called the \emph{prime ideal}, and any element $\beta\in P$ such that $P=\beta O$ is called a \emph{uniformizer} (or a \emph{prime element}) of $K$. For $\qp$ one can choose $\beta=p$, and for $\mathbf F_q((t))$ one can take $\beta=t$.\\
The set $U$ coincides with the \emph{group of units} of $O$. The quotient ring $O/P$ is a finite field. If $q=p^f$ is the number of elements in $O/P$ ($p$: a prime number, $f$: a natural number) and $\beta$ is a uniformizer, then $|\beta|=1/q$, and the range of values of $|\cdot|$ is $\{q^N: N\in\mathbf Z\}$. Further, if $S$ is a complete set of representatives for the residue classes in $O/P$, every non-zero element $x\in K$ can be written uniquely in the form:
\[
x=\beta^{-m}(x_0+x_1\beta+x_2\beta^2+\cdots),
\]
where $m\in\mathbf Z$, $x_j\in S$, $x_0\not\in P$. With $x$ written in this form, we have $|x|=q^m$.

\subsubsection{Characters and Fourier Transform}\label{ft}
We fix a Haar measure $\mu$ on $K$, normalized such that $\mu(O)=1$, and define the Fourier transform $\fo$ on $K$ by 
\[(\fo f)(\xi)=\int_K f(x)\chi(-x\xi)\,dx\,,\]
where $\chi$ is a rank zero\footnote{The rank of a character 
	$\chi$ is defined as the largest integer $r$ such that $\chi\vert_{B_r}\equiv 1$. See \cite{BD15} for explicit construction of such characters in the various cases.} character on $K$, and $dx$ refers to the Haar measure just introduced. 
Any Fourier transform based on a rank zero character is an $L^2$-isometry with respect to the normalized Haar measure (since $\fo\one_O=\one_O$ for any such Fourier transform $\fo$; here and elsewhere $\one$ denotes characteristic function). Thus $\foi=\fo^*$ is given by
\[(\foi f)(x)=(\fo^* f)(x)=\int_K f(y)\chi(xy)\,dy.\]
For the rest of this article $\chi$ will denote a fixed character of rank zero on a local field $K$, and $\fo$ will denote the corresponding Fourier transform.

\subsection{Finite Models}\label{finapprox}
Our object of  study is a version of the Schr\"odinger operator, defined for $\qp$ in the book of Vladimirov, Volovich, Zelenov \cite{VVZ94}, and generalized to an arbitrary local field $K$ by Kochubei in \cite{Koc01}:
\[H=P^\alpha+V\,,\]
regarded as an operator in $L^2(K)$. 
Here $\alpha>0$ \footnote{For a direct analog of the Laplacian one should set $\alpha=2$. However, as is customary in the non-Archimedean setting, one works with an arbitrary $\alpha>0$, since the qualitative behavior of the operator $H$ does not change with $\alpha>0$.}, $P=\foi Q\fo$ where $(Qf)(x)=|x| f(x)$ is the position operator\footnote{Our operator $P$ corresponds to the operator $D$ in \cite{VVZ94} and \cite{Koc01}.}, and $\fo$ is the Fourier transform on $L^2(K)$. $V$ (the potential) is multiplication by a  function: $(Vf)(x)=v(x)f(x)$. We assume $v$ to be non-negative and continuous and that $v(x)\to\infty$ as $|x|\to\infty$. 

The operator $H$ has been thoroughly analyzed (see \cite{VVZ94} for $K=\qp$ and \cite{Koc01} for general $K$): It is self-adjoint on the domain $\{f\in L^2(K):P^\alpha f+Vf\in L^2(K)\}$, has discrete spectrum, and all eigenvalues have finite multiplicity. Our next task is to set up a finite model for this operator.

Keep the above notation, i.e.: $K$ is a local field, $q=p^f$ is the number of elements in the finite field $O/P$, $\beta$ is a uniformizer, and $S$ is a complete set of representatives for $O/P$. For each integer $n$ set $\bn=\beta^{-n}O=\text{ball of radius $q^n$}$. Then $\bn$ is an open, additive subgroup of $K$. For $n>0$ we set $G_n=\bn/\bmn$. Then $G_n$ is a finite group with $q^{2n}$ elements. Since the subgroup $\bmn$ will appear quite frequently, we will often denote it by $H_n$, to emphasize its role as a subgroup. So $H_n=\bmn=\beta^nO=\text{ball of radius $q^{-n}$}$, and $G_n=H_{-n}/H_n$. Each element of $G_n$ has a unique representative of the form $a_{-n}\beta^{-n}+a_{-n+1}\beta^{-n+1}+\dots+a_{-1}\beta^{-1}+ a_0+a_1\beta+\dots+ a_{n-2}\beta^{n-2}+a_{n-1}\beta^{n-1}$, $a_i\in S $. We denote this set by $X_n$, and call it \emph{the canonical set of representatives} for $G_n$; we also give it the group structure coming from its natural identification with $G_n$.

Let again $\mu$ denote the normalized Haar measure on $K$ (cfr.\ \ref{ft}). Since $H_n$ is an open subgroup of $K$, we obtain a Haar measure $\mu_n$ on  $G_n=H_{-n}/H_n$  
by setting $\mu_n(x+H_n)=\mu(x+H_n)=\mu(H_n)=q^{-n}$, for $x+H_n\in G_n$.

So each ``point" $x+H_n$ of $G_n$ has mass $q^{-n}$, and the total mass of $G_n$ is $q^{2n}\cdot q^{-n}=q^n$. For $X_n\simeq G_n$ this means that each $x\in X_n$ has mass $q^{-n}$, and the total mass of $X_n$ is $q^n$.

With this choice of Haar measure on $G_n$ the linear map which sends the characteristic function of the point $x+H_n$ in $G_n$ to the characteristic function of the subset $x+H_n$ of $K$, is an isometric imbedding of $L^2(G_n)$ into $L^2(K)$. We regard $L^2(G_n)$ as a subspace of $L^2(K)$ via this imbedding, and operators on $L^2(G_n)$ are extended to all of $L^2(K)$ by setting them equal to 0 on the orthogonal complement of $L^2(G_n)$ in $L^2(K)$.

We introduce the following subspaces of $L^2(K)$, along with their orthogonal projections:
\begin{itemize}
	\item $\mathcal C_n=\{f\in L^2(K)|\supp(f)\subset B_{n}\}.$ The corresponding orthogonal projection is denoted by $C_n$ and is given by: $C_nf=\one_{B_n}f$.
	\item $\mathcal S_n=\{f\in L^2(K)|\text{$f$ is locally constant of index $q^{-n}$}\}.$ The corresponding orthogonal projection is denoted by $S_n$ and is given by:\\ $(S_nf)(x)=q^n\int_{H_n}f(x+y)\,dy
	=\frac{1}{\mu(H_n)}\int_{H_n}f(x+y)\,dy=\ave(f,n,x)$, where we have introduced the notation $\ave(f,n,x)$ for the average value of $f$ over $x+H_n$.
	\item $\mathcal D_n=\mathcal C_n\cap\mathcal S_{n}.$ 
	The corresponding orthogonal projection is denoted by $D_n$.
\end{itemize}
Note that $L^2(G_n)$ is mapped onto $\mathcal D_n$ via the isometric imbedding mentioned above. Thus $L^2(G_n)$ can be identified with the set of functions on $K$ which have support in $B_n$ and which are invariant under translation by elements of $H_n\,(=B_{-n})$.

Of course, by using the identification $x\in X_n \longleftrightarrow x+H_n\in G_n$, all of the above statements remain valid when $G_n$ is replaced by $X_n$

We now collect the basic facts and conventions for the finite level operators (for details, see \cite{BD15}):
\begin{align*}
	&D_n=C_nS_n=S_nC_n\,.\\
	&\fo\mathcal C_n=\mathcal S_{n},\quad \fo\mathcal S_n=\mathcal C_{n},\text{ and hence }\fo\mathcal D_n=\mathcal D_n\,.\\
	&\fo C_n=S_n\fo,\quad \fo S_n=C_n\fo,\quad \fo D_n=D_n\fo\,.\\
	&\text{Finite Fourier transform $\fo_n$: }\\
	&\qquad(\mathcal{F}_nf)(x)=\qmn\sum_{y\in X_n}f(y)\chi(-xy),\quad x\in X_n,\quad f\in L^2( X_n)\,.\\
	&\fo\vert_{\dn}=\fo_n, \text{ i.e., } \fo_n=\fo D_n=D_n\fo\,.
\end{align*}
\subsubsection{Dynamical Operators at the Finite Level}\label{dynop}
For the finite versions of the dynamical operators we could, as in \cite{BD15}, use their compressions by $D_n$, i.e., $V'_n=D_nVD_n$, $Q'_n=D_nQD_n$, $P'_n=D_nPD_n=\foi_n Q'_n\fo_n$, and $H'_n=D_nHD_n=D_nP^\alpha D_n+V'_n$. However, since our dynamical operators are defined by continuous functions, it will be more convenient to descend to the finite level via the following operator
\beq\label{en}
E_nf=\sum_{y\in X_n}f(y)\one_{y+H_n},\quad f\in C(K).
\eeq
This is a linear idempotent with range $\mathcal D_n\simeq L^2(X_n)$. It is continuous with respect to the topology of uniform convergence on compacta on $C(K)$ (but discontinuous w.r.t.\ the $L_2$-norm on $C(K)\cap L^2(K)$).
Note that $\lim_{n\to\infty}(D_nf)(x)=\lim_{n\to\infty}(E_nf)(x)$ if $f$ is continuous. The finite version of a function $f$ on $K$ can be thought of either as an element of $\mathcal D_n$ according to \eqref{en}, or as a function on the grid $X_n$, where it is simply given by its restriction $f\vert_{X_n}$. We will switch between these two points of view depending on what seems more convenient in a given situation. When working on $X_n$ we will often make no notational distinction between a function on $K$ and its restriction to $X_n$.
For a function of two variables $f\in C(K\times K)$ we similarly have
\beq\label{enen} 
(E_n\otimes E_n)f=\sum_{x,y\in X_n}f(x,y)\one_{(x+H_n)\times(y+H_n)},\quad f\in C(K\times K),
\eeq
which can be thought of as the restriction of $f$ to $X_n\times X_n$.

For the finite versions $Q_n,P_n,H_n$ of the operators $Q,P,H$, we take
\begin{align}\label{finops}\begin{split} 
(Q_nf)(x)&=|x|f(x),\quad f\in L^2(X_n),\quad x\in X_n\\
P_n&=\foi_n Q_n\fo_n\\
(V_nf)(x)&=v_n(x)f(x),\quad v_n=v|_{X_n},\quad f\in L^2(X_n)\\
H_n&=P_n^\alpha + V_n,\quad \alpha>0
\end{split}
\end{align}
Note that the finite operators $Q_n,P_n,H_n$ can also be viewed as operators on $L^2(K)$ via the identification of $L^2(X_n)$ with $\mathcal D_n=D_nL^2(K)$.

\section{\large \bf  Stochastics at the Finite Level}\label{sec:finitelevel}
We start by recalling the connection between Brownian motion and the heat equation in the conventional setting over $\rr$. Here Brownian motion is described by a family of Wiener measures $(W_x)_{x\in\rr}$, which in turn are generated by the probability densities\footnote{We are using self-dual Haar measure $dz/\sqrt{2\pi}$ on $\rr$.} 
$p_t(z)=\frac{1}{\sqrt{2t}}e^{-z^2/4t}$, $z\in\rr$, $t>0$. The relation 
\[
\int_{C([0,\infty):\rr)}f(\omega(t))dW_x(\omega)=\frac{1}{\sqrt{2\pi}}\int_\rr f(y)p_t(x-y)dy
\]
holds for all ``observables" $f$ belonging to a suitable class of functions on $\rr$. The function $u(x,t)=p_t(x)$ is a fundamental solution of the heat equation
\begin{align}\label{heateq}
	\frac{\partial u}{\partial t}(x,t)&=\Delta u(x,t)\\
	\intertext{which by Fourier transform becomes}
	\label{fheateq}
	\frac{\partial\hat u}{\partial t}(\xi,t)&=-\xi^2 \hat u(\xi,t)\\
	\intertext{and so}
	\label{solfheateq}
	\hat p_t(\xi)&=\hat{u}(\xi,t)=e^{-t\xi^2},
\end{align}
taking into account that $p_t(x)$ is a fundamental solution. The $(p_t)_{t>0}$ form a semi-group under convolution, and thus give rise to a semi-group of operators $(T_t)_{t>0}$ by $T_tf=p_t*f$. The infinitesimal generator of $(T_t)_{t>0}$ is the Laplacian $\Delta$ (on a suitable domain), so we can also write $e^{t\Delta}f=p_t*f$.

Over a local field $K$ one still lets $t$ be a positive real parameter, but the role of the Laplacian $\Delta$ is played by the operator $-P^\alpha$ (remember that $\Delta=-P^2$ over $\rr$), and so the heat equation \eqref{heateq} becomes
\begin{align}\label{locheateq}
	\frac{\partial u}{\partial t}(x,t)&=-(P^\alpha u)(x,t),\quad \text{i.e., }\frac{\partial u}{\partial t}(x,t)=-(\foi Q^\alpha\fo u)(x,t)\,,\\
	\intertext{thus}
	\label{locfheateq}
	\frac{\partial \hat u}{\partial t}(\xi,t)&=-|\xi|^\alpha \hat u(\xi,t)\,,\\
	\intertext{giving}
	\label{locsolfheateq}
	\hat{u}(\xi,t)&=e^{-t|\xi|^\alpha}
\end{align}
by a similar normalization as above. In analogy with the real case one now defines 
\beq \label{locdensity}
p_t(x)=(\foi e^{-t|\cdot|^\alpha})(x)=\int_K e^{-t|\xi|^\alpha}\chi(x\xi)\,d\xi.
\eeq 
The $(p_t)_{t>0}$ again form a semi-group under convolution (since clearly $(\hat p_t)_{t>0}$ form a semi-group under multiplication), and $\int_K p_t(x)dx=1$ for all $t>0$ (since $\hat p_t(0)=1$ for all $t>0$). Thus the only thing missing for the $(p_t)_{t>0}$ to generate a Wiener measure as above, is the positivity of the $(p_t)_{t>0}$. And this has been proved by several authors in various settings (see \cite[Ch.\ 4]{Koc01} and references therein, and \cite{Var97}).

For our finite model we pursue the above analogy and define
\beq
\ptn(x)=(\foi_ne^{-t|.|^\alpha})(x),\quad x\in X_n
\eeq
in analogy with \ref{locdensity}. Here we regard
$e^{-t|.|^\alpha}$ as a function on $X_n$ as explained above (cfr.~\ref{dynop}).
We still have
\beq\label{etpanf}
\etpan f=\ptn*f
\eeq
since
\begin{align*}
	(\etpan f)(x)&=(e^{-t\foi_n Q_n^\alpha\fo_n}f)(x)=(\foi_n e^{-tQ_n^\alpha}\fo_n f)(x)\\
	&=(\foi_n(e^{-t|\cdot|^\alpha}\fo_n f))(x)=(\foi_n(e^{-t|\cdot|^\alpha})*f)(x)\\
	&=(\ptn*f)(x),
\end{align*}

where the convolution $*$ now is over $X_n$:
\[
(f*g)(x)=\int_{X_n}f(y)g(x-y)d\mu_n(y)=q^{-n}\sum_{y\in X_n}f(y)g(x-y).
\]
The one-parameter family \ptnt\ is a semi-group under convolution (since clearly $(\hat{p}_{t,n})_{t>0}$ is a multiplicative semi-group), and $\int_{X_n}\ptn(x)dx=1$ for all $n$ and for all $t>0$ (since $\hat p_{t,n}(0)=1$). It remains to show that the \ptn\ are positive.
\begin{lemma}\label{positivity}
	We have $\ptn(x)>0$ for all $x\in X_n$, all $n$ and all $t>0$, hence \ptnt\ defines a probability distribution over $X_n$.
\end{lemma}
\begin{proof}
	Remember that functions in $L^2(X_n)$ can be thought of as functions on $K$ which are supported in $B_n$ and which are locally constant of index $\qmn$. We use that picture here. For example, the function $\xi\to \expta$ is interpreted as the function $\sum_{\xi\in X_n}\expta\one_{\xi+H_n}$ (cfr.~\ref{dynop}).\\ Below we also use the notation $S_i=\{x\in K:|x|=p^i\}=B_i\setminus B_{i-1}$.
	\vspace{-\baselineskip}
	\begin{align*}
		\ptn(x) &= (\foi\hat p_{t,n})(x)=\int_{B_n}e^{-t|\xi|^\alpha}\chi(x\xi)d\xi \\
		&= \int_{B_{-n}} \, d\xi + \sum_{-n+1\leq i \leq n} e^{-tq^{\alpha i}}\int_{S_i}  \chi(x\xi) \, d\xi \\
		&= q^{-n} + \sum_{-n+1\leq i \leq n} e^{-tq^{\alpha i}}\left(\int_{B_i}  \chi(x\xi) \, d\xi - \int_{B_{i-1}}  \chi(x\xi) \, d\xi \right)\\
		&= q^{-n} + \sum_{-n+1\leq i \leq n} e^{-tq^{\alpha i}}\int_{B_i}  \chi(x\xi) \, d\xi - \sum_{-n\leq i \leq n-1} e^{-tq^{\alpha(i+1)}}\int_{B_i}  \chi(x\xi) \, d\xi \\
		&= q^{-n} - e^{-tq^{\alpha(-n+1)}} \int_{B_{-n}}  \chi(x\xi) \, d\xi + e^{-tq^{\alpha n}} \int_{B_n}  \chi(x\xi) \, d\xi \\
		&+ \sum_{-n+1\leq i \leq n-1} (e^{-tq^{\alpha i}} - e^{-tq^{\alpha(i+1)}})\int_{B_i}  \chi(x\xi) \, d\xi \\
		&= q^{-n}(1 - e^{-tq^{\alpha(-n+1)}}) + e^{-tq^{\alpha n}} \int_{B_n}  \chi(x\xi) \, d\xi \\
		&+ \sum_{-n+1\leq i \leq n-1} (e^{-tq^{\alpha i}} - e^{-tq^{\alpha (i+1)}})\int_{B_i}  \chi(x\xi) \, d\xi\,.
	\end{align*}
	The integrals $\int_{B_i}  \chi(x\xi)\,d\xi$ are always non-negative (see \cite[p.\ 42]{VVZ94} for the case $K=\qp$; the same proof works for a general $K$), hence each term is non-negative, and the first is positive, so \ptn\ is positive on $X_n$ for all $t > 0$.
\end{proof}

\section{\large \bf  Convergence of Measures}
From now on we'll be working on a fixed time interval which we will denote by $[0,t]$; a generic time point in $[0,t]$ will be denoted by $s$. We start by recalling the above formulas for the densities (with the time parameter $t$ replaced by $s$):
\begin{align}
	\psn(x)&=\int_{B_n}e^{-s|\xi|^\alpha}\chi(x\xi)\,d\xi \notag\\ \label{ptnsum}
	&= q^{-n}(1 - e^{-sq^{\alpha(-n+1)}}) + e^{-sq^{\alpha n}} \int_{B_n}  \chi(x\xi) \, d\xi \\ 
	&+ \sum_{-n+1\leq i \leq n-1} (e^{-sq^{\alpha i}} - e^{-sq^{\alpha (i+1)}})\int_{B_i}  \chi(x\xi) \, d\xi\notag
\end{align}
\begin{align}
	p_s(x)&=\int_K e^{-s|\xi|^\alpha}\chi(x\xi)\,d\xi \notag\\ \label{ptsum}
	&=\sum_{i \in \mathbf{Z}} (e^{-sq^{\alpha i}} - e^{-sq^{\alpha(i+1)}})\int_{B_i}  \chi(x\xi) \, d\xi\,.
\end{align}
We now introduce the space $D[0,t]$ of Skorokhod functions. These are the functions  defined on the interval $[0,t]$ with values in $K$ which satisfy the following two criteria:
\begin{enumerate}
	\item For each $s\in(0,t)$, $f(s \pm 0)$ exist; $f(0 + 0)$ and $f(t-0)$ exist.
	\item $f(s + 0) = f(s)$ for all $s \in[0, t)$, and $f(t) = f(t - 0)$.
\end{enumerate}
We will use the densities $\psn$ to construct, for each $n$ and for each $a\in X_n$, a probability measure $\mathbf{P}_a^n$ on the space $D[0,t]$, and subsequently show that these measures converge weakly to the measure $\mathbf{P}_a$ on $D[0,t]$ which is constructed from the densities $p_s$. The measure $\Pn_a$ will give full measure to the paths which take values in the grid $X_n$. To achieve all of this we need a few lemmas.  
\begin{lemma}\label{ptunifbdd}
	The \psns\ are uniformly bounded, that is, for each $s\in (0,t]$ there is a constant $B_s$ such that
	\begin{equation}
	||\psn||_{\infty} < B_s \notag
	\end{equation}
	for all $n$.
\end{lemma}

\begin{proof}
	By \eqref{ptnsum} we have
	\begin{align*}
		\psn(x) &= q^{-n}(1 - e^{-sq^{\alpha(-n+1)}}) + e^{-sq^{\alpha n}} \int_{B_n}  \chi(x\xi) \, d\xi \\
		&+ \sum_{-n+1\leq i \leq n-1} (e^{-sq^{\alpha i}} - e^{-sq^{\alpha (i+1)}})\int_{B_i}  \chi(x\xi) \, d\xi.
	\end{align*}
	The first and second term go to $0$ uniformly when  $n \rightarrow \infty$ since 
	$| \int_{B_n}  \chi(x\xi) \, d\xi| \leq q^n$. The third term is bounded by  $\sum_{i \in \mathbf{Z}} (e^{-sq^{\alpha i}} - e^{-sq^{\alpha (i+1)}})\int_{B_i}  \chi(x\xi) \, d\xi$, and the latter is uniformly bounded according to \cite[Lemma 2, Sec.\ 4, proof]{Var97}.
\end{proof}
\begin{lemma}\label{uniformConvDensity}
	\psn(x)\ converges uniformly to $p_s(x)$ on compact sets.
\end{lemma}
\begin{proof}
	Let $E$ be a compact subset of $K$ and choose $n_0$ so that $E\subset B_n$ for $n\geq n_0$. Then for $x\in E$ and $n\geq n_0$ we have:
	\begin{align*}
		|p_s(x)-\psn(x)| &\leq q^{-n}(1 - e^{-sq^{\alpha(-n+1)}}) + e^{-sq^{\alpha n}} \int_{B_n}  \chi(x\xi) \, d\xi \\
		&+ \sum_{\substack{i \leq -n\\ i \geq n}} (e^{-sq^{\alpha i}} - e^{-sq^{\alpha(i+1)}})\int_{B_i}  \chi(x\xi) \, d\xi\,.
	\end{align*}
	The first terms goes to 0 as $n \rightarrow \infty$, and so does the second since $|\int_{B_n}  \chi(x\xi) \, d\xi|\leq q^n$. For the third term we again take advantage of an estimate from \cite[Lemma 2, Sec.\ 4, proof]{Var97}, this time writing it out more explicitly:
	\begin{multline*}
		\sum_{\substack{i \leq -n\\ i \geq n}} (e^{-sq^{\alpha i}} - e^{-sq^{\alpha(i+1)}})\int_{B_i}  \chi(x\xi) \, d\xi 
		\leq \sum_{\substack{i \leq -n\\ i \geq n}} (e^{-sq^{\alpha i}} - e^{-sq^{\alpha (i+1)}})q^i\\
		\leq \sum_{\substack{i \leq -n\\ i \geq n}} s\int_{q^{i\alpha}}^{q^{\alpha(i+1)}} e^{-sy}y^{1/\alpha} \, dy
		= s\int_{[0,q^{(-n+1)\alpha}] \cup [q^{n\alpha},\infty)} e^{-sy}y^{1/\alpha} \, dy\,.
	\end{multline*}
	The last term goes to 0 (being the tail of a convergent integral), so $\psn$ converges pointwise to $p_s$. Since the estimates are independent of $x$, we have uniform convergence on $E$.	
	
\end{proof}	

We now start the construction of the measures $\Pn_a$. Pick a point $a\in X_n$, fix $N$ time points $0\leq t_1<t_2<\dots<t_N\leq t$, and for each $i=1,\dots N$, pick a Borel subset $J_i$ of $K$. We define a measure $\Pn_a$ on the cylinder sets $\{\omega:[0,t]\to K:\omega(t_i)\in J_i\}$ by 
\begin{multline}\label{cyl}
	\mathbf{P}_a^n(\omega(t_i) \in J_i)\\= \sum_{b_i \in J_i \cap X_n,\, 1 \leq i \leq N} p_{t_1,n}(b_1-a)\cdots p_{t_N-t_{N-1},n}(b_N-b_{N-1}) q^{-nN}.
\end{multline} 
By Kolmogorov's Extension Theorem \cite[Thm.\ 2.1.5]{Oks98}, $\mathbf{P}_a^n$ has a unique extension to a probability measure on $\Omega[0,t]$, the space of all functions $\omega : [0,t] \rightarrow K$, equipped with the $\sigma$-algebra generated by the cylinder sets. To get a probability measure on $D[0,t]$, equipped with the Borel sets coming from the Skorokhod topology, we need to check the \v Centsov criterion, which says: If there are constants $c,d,e,C>0$ such that 
\beq
\label{cen} E_{\Pn_a}(|Y_{t_1} - Y_{t_2}|^c|Y_{t_2} - Y_{t_3}|^d) \leq C|t_1 - t_3|^{1+e}
\eeq
for all $0\leq t_1<t_2<t_3\leq t$, then there is a unique measure on $D[0,t]$ which satisfies the condition \eqref{cyl}. Here $E_{\Pn_a}$ denotes the expectation w.r.t.\ the measure $\Pn_a$, and $Y_s$ denotes the random variable $Y_s(\omega)=\omega(s)$, $\omega\in\Omega[0,t]$, $s\in[0,t]$. The random variables $Y_s$ define a process with independent increments with respect to each of the measures $\Pn_a$.

\begin{proposition}\label{Centsovprop}
	Let $k$ be a real number with $0<k<\alpha$, and pick time points $0\leq t_1<t_2<t_3\leq t$. Then there is a constant $D_k>0$ such that
	\beq\label{kpowerestimate}
	E_{\Pn_a}(|Y_{t_1} - Y_{t_2}|^k|Y_{t_2} - Y_{t_3}|^k) \leq D_k|t_1 - t_3|^{2k/\alpha}.
	\eeq
	If also $k>\alpha/2$, then \v Centsov's condition \eqref{cen} is satisfied.
\end{proposition}

\begin{proof}
	Using the point $a=0$ in $X_n$, we have
	
	\begin{multline*}
		E_{\mathbf{P}_0^n}(|Y_s|^k) = \int_{\Omega[0,t]} |Y_s(\omega)|^k \,d\mathbf{P}_0^n(\omega) =  \int_{K} |x|^k \, d\mathbf{P}_0^n \circ Y_s^{-1}(x)\\
		= \sum_{x \in X_n} \int_{ \{x\}} |x|^k \, d\mathbf{P}_0^n \circ Y_s^{-1}(x) + \int_{K \setminus X_n} |x|^k \, d\mathbf{P}_0^n \circ  Y_s^{-1}(x) \\
		= \sum_{x\in X_n} |x|^k \psn(x) q^{-n} = \sum_{x \in X_n, x\neq 0} |x|^k \psn(x) q^{-n}.
	\end{multline*}
	Using the expression \eqref{ptnsum} for $\psn$, we get
	\begin{align*}
		&E_{\mathbf{P}_0^n}(|Y_s|^k)\\
		&=q^{-n} \sum_{x \in X_n, x\neq 0} |x|^k \bigg{(} q^{-n}(1 - e^{-sq^{\alpha (-n+1)}}) + e^{-sq^{\alpha n}} \int_{B_n}  \chi(x\xi ) \, d\xi \\
		&+ \sum_{-n+1\leq i \leq n-1} (e^{-sq^{\alpha i}} - e^{-sq^{\alpha (i+1)}})\int_{B_i}  \chi(x\xi) \, d\xi \bigg{)} \\
		&= q^{-n} \sum_{x \in X_n, x\neq 0} |x|^k \bigg{(} q^{-n}(1 - e^{-sq^{\alpha (-n+1)}}) 
		\\&+ \sum_{-n+1\leq i \leq n-1} (e^{-sq^{\alpha i}} - e^{-sq^{\alpha (i+1)}})\int_{B_i}  \chi(x\xi) \, d\xi \bigg{)} \\
		&= q^{-n} (1 - e^{-sq^{\alpha (-n+1)}})\int_{B_n \setminus B_{-n}}|x|^k \, dx \\&+ \int_{B_n \setminus B_{-n}} |x|^k\sum_{-n+1\leq i \leq n-1} (e^{-sq^{\alpha i}} - e^{-sq^{\alpha (i+1)}})\int_{B_i}  \chi(x\xi) \, d\xi dx \\
		&\leq q^{-n}(1 - e^{-sq^{\alpha (-n+1)}}) q^n q^{nk} \\
		&+ \sum_{-n+1\leq i \leq n-1} (e^{-sq^{\alpha i}} - e^{-sq^{\alpha (i+1)}}) \int_{B_{-i}} |x|^k \int_{B_i}  \chi(x\xi) \, d\xi dx \\
		&\leq q^{-n}(1 - e^{-sq^{\alpha (-n+1)}}) q^n q^{nk} 
		+ \sum_{-\infty< i < \infty} (e^{-sq^{\alpha i}} - e^{-sq^{\alpha (i+1)}}) q^{-ik}q^{-i}q^{i} \\
		&= q^{-n}(1 - e^{-sq^{\alpha (-n+1)}}) q^n q^{nk} 
		+ \sum_{-\infty< i < \infty} (e^{-sq^{\alpha i}} - e^{-sq^{\alpha (i+1)}}) q^{-ik}\,.\\
		\intertext{At this point we again invoke an inequality by Varadarajan \cite{Var97}[Lemma 2, Sec.\ 4, proof], which in our setting translates to $\sum_{-\infty< i < \infty} (e^{-sq^{\alpha i}} - e^{-sq^{\alpha (i+1)}}) q^{-ik}\leq A_k s^{k/\alpha}$ for some constant $A_k$ which is independent of $n,s$. The chain of inequalities then continues as (with $B_k,C_k$ some other constants which are independent of $n,s$)}
		&\leq (1 - e^{-sq^{\alpha (-n+1)}}) q^{nk} + A_k s^{k/\alpha}\leq  s q^{-n\alpha}q^\alpha q^{nk}+ A_k s^{k/\alpha}\\&=	s q^\alpha q^{-n(\alpha-k)}+ A_k s^{k/\alpha}\leq s q^\alpha + A_k s^{k/\alpha}\leq B_k  s^{k/\alpha} + A_k s^{k/\alpha}\leq C_k s^{k/\alpha}\,,
	\end{align*}
	where we have used that $\alpha-k>0$, and that over the finite interval $[0,t]$, we can make $ s q^\alpha\leq B_k s ^{k/\alpha}$ for a suitable $B_k$. To sum it up, we have shown that
	\beq\label{momentineq}
	E_{\mathbf{P}_0^n}(|Y_s|^k)\leq C_k s^{k/\alpha}
	\eeq
	for some constant $C_k$ which is independent of $n,s$. Using that the process $Y_t$ has stationary increments and that $Y_0=0$ with $\mathbf P^n_0$-probability $1$, we get
	\begin{multline*}
		E_{\mathbf{P}_0^n}(|Y_{t_2}-Y_{t_1}|^k|Y_{t_3}-Y_{t_2}|^k)=
		E_{\mathbf{P}_0^n}(|Y_{t_2-t_1}-Y_{0}|^k|Y_{t_3-t_2}-Y_{0}|^k)\\
		=E_{\mathbf{P}_0^n}(|Y_{t_2-t_1}|^k|Y_{t_3-t_2}|^k)
		\leq
		(E_{\mathbf{P}_0^n}(|Y_{t_2-t_1}|^{2k}))^{1/2}
		(E_{\mathbf{P}_0^n}(|Y_{t_3-t_2}|^{2k}))^{1/2}\\
		\stackrel{\eqref{momentineq}}{\leq}C_{2k}(t_2-t_1)^{k/\alpha}(t_3-t_2)^{k/\alpha}<C_{2k}(t_3-t_1)^{2k/\alpha}\,.
	\end{multline*}
	
	Noticing that
	\[
	E_{\mathbf{P}_a^n}(|Y_{t_2}-Y_{t_1}|^k|Y_{t_3}-Y_{t_2}|^k)=E_{\mathbf{P}_0^n}(|Y_{t_2}-Y_{t_1}|^k|Y_{t_3}-Y_{t_2}|^k)
	\]
	for any $a\in X_n$ (since only differences between the $Y_{t_i}$ occur), we finally get
	\begin{equation}
	\label{CentsovInequ}
	E_{\mathbf{P}_a^n}(|Y_{t_2}-Y_{t_1}|^k|Y_{t_3}-Y_{t_2}|^k)\leq C_{2k} (t_3-t_1)^{2k/\alpha},
	\end{equation}
	for $k<\alpha$. So with  $D_k=C_{2k}$, \eqref{kpowerestimate} holds. If also $k >\alpha/2$, the \v Centsov criterion holds.
\end{proof} 	

\subsection{Convergence of Unconditioned Measures}
The concept of weak convergence of probability measures will play an important role in this article.  
\begin{definition}[Weak Convergence]\label{defweakconv}
	Let $(\mathbf P_n)$ and $\mathbf P$ be probability measures on a metric space $M$. We say that the sequence $(\mathbf P_n)$ converges weakly to $\mathbf P$ -- written $\mathbf P_n\Rightarrow \mathbf P$ -- if $\mathbf P_n(f) \to \mathbf P(f)$ for all bounded, continuous real functions $f$ on $M$.
\end{definition}
For several equivalent definitions, see \cite[Thm.\ 2.1]{Bil99} (``Portmanteau Theorem").

Let $a_n\in X_n$, $a\in K$ be such that $a_n\to a$ as $n\to\infty$. We wish to prove that $\mathbf{P}_{a_n}^n \Rightarrow \mathbf{P}_a$ as $n \rightarrow \infty$. To do this we will use the following theorem from \cite{Var94} (see also \cite[Theorem 13.1]{Bil99}):
\begin{theorem}[Theorem 2, Ch.\ 11, in \cite{Var94}]
	Suppose that $\mathbf{P}_m, \mathbf{P}$ are probability measures on $D[0,t]$ such that
	\begin{itemize}
		\item\mbox{}\vspace{-\baselineskip}
		\beq \mathbf{P}_m^{t_1,...,t_N} \Rightarrow \mathbf{P}^{t_1,...,t_N}\text{ for all $t_1,...,t_N$ in $[0,t]$.}\label{findimconv}
		\eeq
		\item  There are constants $c,d,e,C > 0$ such that for all $n$ and\\ $0 \leq t_1 < t_2 < t_3 \leq t,$
		\beq
		\label{tight2}E_{\mathbf{P}_m}(|Y_{t_2}-Y_{t_1}|^c|Y_{t_3}-Y_{t_2}|^d) \leq C (t_3-t_1)^{1+e}.
		\eeq
	\end{itemize}
	Then $\mathbf{P}_m \Rightarrow \mathbf{P}$.
\end{theorem}
By equation (\ref{CentsovInequ}) the condition \eqref{tight2} is satisfied if $c=d=k, 1+e=2k/\alpha$ and $\alpha/2<k<\alpha$. To prove \eqref{findimconv} we can use the following theorem.
\begin{theorem}[Thm.\ 2.2 in \cite{Bil99}]
	\label{piSystemConv}
	Let $\mathbf{P}$, $(\mathbf{P}_m)_{m=1}^\infty$, be probability measures on $D[0,t]$, and suppose that
	\begin{itemize}
		\item $\mathcal{A}_P$ is a $\pi$-system\footnote{A class of subsets is a $\pi$-system
			if it is closed under the formation of finite intersections.}
		\item Every open set is a countable union of elements in $\mathcal{A}_P$.
	\end{itemize}
	If $\mathbf{P}_m (A) \rightarrow \mathbf{P}(A)$ for all $A \in \mathcal{A}_P$, then $\mathbf{P}_m \Rightarrow \mathbf{P}$.
\end{theorem}
In $K$ the set of all balls is a basis for the topology. In $K^N$, the set of all products of balls, $A_1 \times \cdots \times A_N$, is a basis for the topology. This set is also closed under finite intersections, so we can use Theorem \ref{piSystemConv} to prove convergence of the finite dimensional distributions. So fix a set $A_1 \times \cdots \times A_N$. Let $a_n\in X_n \rightarrow a\in K$ as $n\rightarrow \infty$. We wish to prove that $\mathbf{P}_{a_n}^{n,t_1,...,t_N}(A_1 \times \cdots \times A_N) \rightarrow \mathbf{P}_a^{t_1,...,t_N}(A_1 \times \cdots \times A_N)$ as $n \rightarrow \infty$. We have
\begin{multline*}
	\mathbf{P}_{a_n}^{n,t_1,...,t_N}(A_1 \times \cdots \times A_N)\\= \sum_{b_i \in A_i, 1 \leq i \leq N} p_{t_1,n}(b_1-a_n)\cdots p_{t_N-t_{N-1},n}(b_N-b_{N-1}) q^{-nN}.
\end{multline*}
Let $n$ be large enough so that the balls $A_1,...,A_N$ all have radius larger than $q^{-n}$. Then
\begin{multline*}
	\mathbf{P}_{a_n}^{n,t_1,...,t_N}(A_1 \times \cdots \times A_N)\\
	= \sum_{b_i \in A_i, 1 \leq i \leq N} p_{t_1,n}(b_1-a_n)\cdots p_{t_N-t_{N-1},n}(b_N-b_{N-1}) q^{-nN}\\
	= \sum_{b_i \in A_i, 1 \leq i \leq N-1}\int_{A_N} p_{t_1,n}(b_1-a_n)\cdots p_{t_N-t_{N-1},n}(x_N-b_{N-1}) q^{-n(N-1)} \, dx_N \\
	=  \int_{A_1}\cdots \int_{A_N} p_{t_1,n}(x_1-a_n)\cdots p_{t_N-t_{N-1},n}(x_N-x_{N-1}) \, dx_N \cdots dx_1\,. \\
\end{multline*}
We also have that
\begin{multline*}
	\mathbf{P}_a^{t_1,...,t_N}(A_1 \times \cdots \times A_N)\\ =  \int_{A_1}\cdots \int_{A_N} p_{t_1}(x_1-a)\cdots p_{t_N-t_{N-1}}(x_N-x_{N-1}) \, dx_N \cdots dx_1\,.
\end{multline*}
When $a_n \rightarrow a$,
\begin{align*}
	&\int_{A_1}\cdots \int_{A_N} p_{t_1,n}(x_1-a_n)\cdots p_{t_N-t_{N-1},n}(x_N-x_{N-1}) \, dx_N \cdots dx_1 \\
	& \rightarrow \int_{A_1}\cdots \int_{A_N} p_{t_1}(x_1-a)\cdots p_{t_N-t_{N-1}}(x_N-x_{N-1}) \, dx_N \cdots dx_1
\end{align*}
by Lemma \ref{ptunifbdd}, and since the probability densities converge uniformly on compact sets. Thus $\mathbf{P}_{a_n}^{n,t_1,...,t_N} \Rightarrow \mathbf{P}_a^{t_1,...,t_N}$ and hence $\mathbf{P}_{a_n}^n \Rightarrow \mathbf{P}_a$. We have proved:
\begin{theorem}[Weak Convergence of Unconditioned Measures]
	Let $a_n\in X_n$, $a\in K$ be such that $a_n\to a$ as $n\to\infty$. Then 
	\[\mathbf{P}_{a_n}^n \Rightarrow \mathbf{P}_a\text{ as }n \rightarrow \infty\,,\] 
	where, we recall, $\Rightarrow$ denotes weak convergence of measures.
\end{theorem}

\subsection{Convergence of Conditioned Measures}
Let $a,b\in X_n$. The conditioned measure $\mathbf{P}^n_{a,b,t}$ of a Borel set $A\subset D[0,t]$ is defined by\footnote{Here and in the following we use the probabilist's notation for sets: $(\omega(t) = b)$ is a shortcut notation for the set $\{\omega:\omega(t) = b\}$. More generally, for time points $0\leq t_1<\cdots< t_N\leq t$ and Borel sets $J_i,i=1,\dots N$, the notation $(\omega(t_i)\in J_i)$ means $\{\omega:\omega(t_i)\in J_i,i=1\dots N\}$. }
\begin{equation}
\mathbf{P}^n_{a,b,t}(A) = \frac{\mathbf{P}^n_{a} (A \cap (\omega(t) = b))}{\mathbf{P}^n_{a}(\omega(t) = b)}.
\end{equation}
In this subsection we wish to prove the following theorem:
\begin{theorem}[Weak Convergence of Conditioned Measures]\label{uniformConvergenceMeasures}
	\label{WeakConvergenceMeasures}
	If $a_n\in X_n \rightarrow a\in K$ and $b_n\in X_n \rightarrow b\in K$, then $\mathbf{P}^n_{a_n,b_n,t} \Rightarrow \mathbf{P}_{a,b,t}$. The convergence is uniform when $(a,b)$ varies in compact subsets of $K\times K$.
\end{theorem}
The proof of this theorem will occupy the remainder of this subsection. 	We first prove the statement about weak convergence. To do this we first prove it for the corresponding finite dimensional distributions.
\begin{proposition}
	Let $a_n,b_n,a,b$ be as in the theorem, and pick time points $0<t_1<\dots<t_N<t$ in $[0,t]$. Then
	\[
	\mathbf{P}^{n,t_1,\dots,t_N}_{a_n,b_n,t}\Rightarrow  \mathbf{P}^{t_1,\dots,t_N}_{a,b,t}\,.
	\]
\end{proposition}
\begin{proof}
	Let $J_i$, $i=1,\dots,N$, be balls in $K$. Then by definition
	\begin{equation*}
		\mathbf{P}^n_{a_n,b_n,t}(\omega(t_i) \in J_i) = \frac{\mathbf{P}^n_{a_n} ((\omega(t_i) \in J_i) \cap (\omega(t) = b_n))}{\mathbf{P}^n_{a_n}(\omega(t) = b_n)}.
	\end{equation*}
	Here the denominator is equal to $p_{t,n}(b_n-a_n)q^{-n}$.  For the numerator we have
	\begin{multline*}
		\mathbf{P}^n_{a_n} ((\omega(t_i) \in J_i) \cap (\omega(t) = b_n))\\ = \int_{J_1}\cdots \int_{J_N} p_{t_1,n}(x_1-a_n)\cdots p_{t-t_N,n}(b_n-x_N)q^{-n} \, dx_N \cdots dx_1\,,
	\end{multline*}
	so
	\begin{align*}
		&\mathbf{P}^n_{a_n,b_n,t}(\omega(t_i) \in J_i)\\
		&=\frac{ \int_{J_1}\cdots \int_{J_N} p_{t_1,n}(x_1-a_n)\cdots p_{t-t_N,n}(b_n-x_N) \, dx_N \cdots dx_1}{p_{t,n}(b_n-a_n)}\\
		&\to \frac{ \int_{J_1}\cdots \int_{J_N} p_{t_1}(x_1-a)\cdots p_{t-t_N}(b-x_N) \, dx_N \cdots dx_1}{p_t(b-a)}\\
		&=\mathbf{P}_{a,b,t}(\omega(t_i) \in J_i)\,,
	\end{align*}
	where we have used Lemma~\ref{uniformConvDensity}.
	From Theorem~\ref{piSystemConv} it now follows that $\mathbf{P}^{n,t_1,,,t_N}_{a_n,b_n,t} \Rightarrow \mathbf{P}^{t_1,,,t_N}_{a,b,t}$.
\end{proof}
To finish the proof that $\mathbf{P}_{a_n,b_n,t}^n \Rightarrow \mathbf{P}_{a,b,t}$, we invoke a result from Billingsley \cite{Bil99}. To state it we need a concept which for Skorokhod functions plays the role of the modulus of continuity:
\begin{equation}
m(\omega:\delta) = \sup_{\substack{s_1 < s < s_2\\0<s_2-s_1 < \delta}} \min \{ |\omega(s_2)-\omega(s)|, |\omega(s)-\omega(s_1)|\}.
\end{equation}
\begin{theorem}[Thm.\ 13.1 in \cite{Bil99}]
	\label{TightnessCond}
	Let $\mathbf{P}$, $(\mathbf{P}_k)_{k=1}^\infty$, be probability measures on $D[0,t]$. If $\mathbf{P}_k^{t_1,,,t_N} \Rightarrow \mathbf{P}^{t_1,,,t_N}$ as $k\to\infty$ for all finite sets of time points $t_1,,,t_N$ and if for every $\eta > 0$
	\begin{equation*}
		\lim_{\delta \rightarrow 0} \mathbf{P}_k(\{\omega : m(\omega:\delta) > \eta \}) = 0
	\end{equation*}
	uniformly in $k$, then $\mathbf{P}_k \Rightarrow \mathbf{P}$ as $k\to\infty$.
\end{theorem}
What is left is to prove is that for every $\eta > 0$,
\begin{equation*}
	\lim_{\delta \rightarrow 0} \mathbf{P}^n_{a_n,b_n,t}(\{\omega : m(\omega:\delta) > \eta \}) = 0
\end{equation*}
uniformly in $n$. To do this we will follow \cite{DVV94}. The idea is to bound the conditioned measures by the unconditioned measures, and use that the latter are tight.

Define for each $\delta, \eta > 0$
\begin{equation}
A(\delta, \eta) = \{\omega : m(\omega:\delta) > \eta\}.
\end{equation}
Also define $m_1$ and $m_2$ to be the analogues of $m$ on the time intervals $[0,3t/4]$ and $[t/4,t]$, respectively.

If $\delta < t/2$, then $s_1$ and $s_2$ are in the same time interval, so
\begin{equation*}
	m(\omega: \delta) = \max\{m_1(\omega: \delta),m_2(\omega: \delta)\}.
\end{equation*}
With $A_j(\delta,\eta)= \{\omega : m_j(\omega:\delta) > \eta\}$ for $j=1,2$, we have
\begin{equation*}
	A(\delta,\eta) = A_1(\delta,\eta) \cup A_2(\delta,\eta).
\end{equation*}
Then it is enough to prove that for every $\eta > 0$,
\begin{equation*}
	\lim_{\delta \rightarrow 0} \mathbf{P}^n_{a_n,b_n,t}(A_j(\delta,\eta)) = 0,
\end{equation*}
uniformly in $n$ for $j=1,2$. We will first prove it for $j=1$ and prove the case $j=2$ by time reflection.

By definition
\begin{equation}
\label{CondA1}
\mathbf{P}^n_{a_n,b_n,t}(A_1) = \frac{\mathbf{P}^n_{a_n} (A_1 \cap (\omega(t) = b_n))}{\mathbf{P}^n_{a_n}(\omega(t) = b_n)}.
\end{equation}
The denominator is equal to $p_{t,n}(b_n-a_n)q^{-n}$. For the numerator we have, for large enough $n$,
\begin{multline*}
	\mathbf{P}^n_{a_n} (A_1 \cap (\omega(t) = b_n)) = \sum_{x \in X_n} \mathbf{P}^n_{a_n} (A_1 \cap (\omega(3t/4) = x) \cap (\omega(t) = b_n))\\
	= \sum_{x \in X_n} \mathbf{P}^n_{a_n} (A_1 \cap (\omega(3t/4) = x) \cap (\omega(t) - \omega(3t/4) = b_n -x))\\
	= \sum_{x \in X_n} \mathbf{P}^n_{a_n} (A_1 \cap (\omega(3t/4) = x)) \mathbf{P}^n_{a_n} (\omega(t) - \omega(3t/4) = b_n -x)
\end{multline*}
by independent increments.
Furthermore, we have the equality
\[
\mathbf{P}^n_{a_n} (\omega(t)-\omega(3t/4)
= b_n-x)=\mathbf{P}_0^n(\omega(t/4) = b_n-x),\]
which follows from the following calculation 
\begin{multline*}
	\mathbf{P}^n_{a_n} (\omega(t) - \omega(3t/4) = b_n -x)\\ = \sum_{y \in X_n} \mathbf{P}^n_{a_n} ((\omega(3t/4) = y)\cap (\omega(t) = y+b_n-x))\\
	= \sum_{y \in X_n} p_{3t/4,n}(y-a_n)p_{t/4,n}(b_n-x)q^{-2n}=p_{t/4,n}(b_n-x)q^{-n}
	\sum_{y \in X_n} p_{3t/4,n}(y-a_n)q^{-n}\\
	=p_{t/4,n}(b_n-x)q^{-n}=\mathbf{P}_0^n(\omega(t/4) = b_n-x).
\end{multline*}
So
\begin{align}\label{pt4nineq}
	\begin{split}
		&\mathbf{P}^n_{a_n} (A_1 \cap (\omega(t) = b_n)) \\&= \sum_{x \in X_n} \mathbf{P}^n_{a_n} (A_1 \cap (\omega(3t/4) = x)) \mathbf{P}_0^n(\omega(t/4) = b_n-x)\\
		&\leq \sum_{x \in X_n} \mathbf{P}^n_{a_n} (A_1 \cap (\omega(3t/4) = x)) \sup_{z \in X_n} p_{t/4,n}(z) q^{-n}
		\\&= \mathbf{P}^n_{a_n} (A_1) \sup_{z \in X_n} p_{t/4,n}(z) q^{-n}.
	\end{split}
\end{align}
Putting this back into equation (\ref{CondA1}) we get
\begin{equation*}
	\mathbf{P}^n_{a_n,b_n,t}(A_1) \leq \frac{\mathbf{P}^n_{a_n} (A_1) \sup_{z \in X_n} p_{t/4,n}(z)}{p_{t,n}(b_n-a_n)}.
\end{equation*}
The denominator	$p_{t,n}(b_n-a_n)$ converges to $p_t(b-a)>0$, so there exists a $\gamma >0$ such that $p_{t,n}(b_n-a_n) \geq \gamma$ for large $n$. Also, there is a $\gamma'>0$ such that $\sup_{z \in X_n} p_{t/4,n}(z) \leq \gamma'$, so
\begin{equation*}
	\mathbf{P}^n_{a_n,b_n,t}(A_1) \leq \frac{\gamma'}{\gamma}\mathbf{P}^n_{a_n} (A_1).
\end{equation*}
The measures $\mathbf{P}^n_{a_n}$ are tight, so -- by \cite[Thm.\ 13.2]{Bil99} and the discussion following it -- we have, for every $\eta > 0$,
\begin{equation*}
	\lim_{\delta \rightarrow 0} \mathbf{P}^n_{a_n,b_n,t}(A_1(\delta,\eta)) \leq \frac{\gamma'}{\gamma} \lim_{\delta \rightarrow 0} \mathbf{P}_{a_n}^n(A_1(\delta,\eta)) =0
\end{equation*}
uniformly in $n$. This proves the statement for $j=1$.\\ 	
To deal with the case $j=2$, we define an operation of time reflection on $D[0,t]$ by
\begin{equation}
\omega^{*}(s) = \omega(t-s-0), \quad 0\leq s<t
\end{equation}
and $\omega^{*}(t)=\omega(0)$. Time reflection is an involutive Borel transformation on $D[0,t]$. At the level of measures we define, for any probability measure $\mathbf{P}$ on $D[0,t]$, the time reflected probability measure $\mathbf{P}^*$ by $\mathbf{P}^{*}(E) = \mathbf{P}(E^{*})$. With this definition
\begin{equation}
(\mathbf{P}^{n}_{a_n,b_n,t})^* = \mathbf{P}^{n}_{b_n,a_n,t}
\end{equation}
This comes from the fact that if $s' <s$, then
\begin{equation*}
	\mathbf{P}^{n}_{a_n,b_n,t}(\omega : |Y_s-Y_{s'}| > \epsilon) = \mathbf{P}^{n}_{0,b_n-a_n,t}(\omega : |Y_{s-s'}| > \epsilon) \rightarrow 0
\end{equation*}
as $s'$ goes to $s$ from below. So $Y_{s'}$ converges to $Y_s$ in measure, but it also converges to $Y_{s^{-}}$ in measure, which shows that $Y_s=Y_{s^{-}}$ almost everywhere. Since this proves that a path is left continuous  with probability one at any given time point, the measures $(\mathbf{P}^{n}_{a_n,b_n,t})^*$ and $\mathbf{P}^{n}_{b_n,a_n,t}$ coincide on cylinder sets, hence on all Borel sets.\\
Now define $m_1^{*}$ as the same as $m_1$ except that $\omega(s)$ is replaced by $\omega(s-0)$. Then
\begin{equation*}
	A_2(\delta, \eta)^{*} = \{\omega : m_1^{*}(\omega:\delta) > \eta\}
\end{equation*}
and
\begin{equation*}
	\{\omega : m_1^{*}(\omega:\delta) > \eta\} \subset \{\omega : m_1(\omega:2\delta) > \eta\}.
\end{equation*}
This gives 
\begin{equation*}
	\mathbf{P}^{n}_{a_n,b_n,t}(A_2(\delta, \eta)) = \mathbf{P}^{n}_{b_n,a_n,t}(A_2(\delta, \eta)^{*}) \leq \mathbf{P}^{n}_{b_n,a_n,t}(A_1(2\delta, \eta)),
\end{equation*}
and hence
\begin{equation*}
	\lim_{\delta \rightarrow 0} \mathbf{P}^n_{a_n,b_n,t}(A_2(\delta,\eta)) = 0
\end{equation*}
for every $\eta > 0$, uniformly in $n$.
Thus for every $\eta > 0$,
\begin{equation*}
	\lim_{\delta \rightarrow 0} \mathbf{P}^n_{a_n,b_n,t}(A(\delta,\eta)) = 0
\end{equation*}
uniformly in $n$, and by Theorem \ref{TightnessCond}, we get that $\mathbf{P}^n_{a_n,b_n,t} \Rightarrow \mathbf{P}_{a,b,t}$, and we have proved the first part of Theorem~\ref{WeakConvergenceMeasures}.

For the second part, let $g$ be a bounded continuous function on $D[0,t]$, and consider the functions 
\begin{align*}h_n(x,y)&= \sum_{(u,v)\in X_n\times X_n}\int_{D[0,t]} g(\omega) d\mathbf{P}^n_{u,v,t}(\omega)\mathbf 1_{(u+H_n)\times(v+H_n)}(x,y)\\ h(x,y) &= \int_{D[0,t]} g(\omega) d\mathbf{P}_{x,y,t}(\omega)
\end{align*} for $(x,y)\in K\times K$. The first part of the theorem tells us that $h_n$ converges continuously to $h$. On a compact subset $A$ of $K\times K$, this implies uniform convergence on $A$.

This completes the proof of Theorem~\ref{WeakConvergenceMeasures}.

We end this section by a theorem on the support of the measures $\mathbf{P}_a^n$.
\begin{theorem}
	For each $a\in X_n$ the measure $\mathbf{P}_a^n$ gives full measure to the paths on the grid, that is,
	\begin{equation*}
		\mathbf{P}_a^n(\omega : \omega(s) \in X_n, \forall s \in [0,t]) = 1
	\end{equation*}
\end{theorem}

\begin{proof}
	By definition of $\mathbf{P}_a^n$, we have $\mathbf{P}_a^n(\omega(t_i) \in X_n, 1 \leq i \leq N) = 1$ for any finite set of time points $t_1,\dots,t_N$ in $[0,t]$. Now take an increasing sequence of sets of finitely many time points $F_i$ such that $\bigcup_{i=1}^{\infty} F_i = \mathbf{Q} \cap [0,t]$. Then
	\begin{equation*}
		\mathbf{P}_a^n(\omega : \omega(s) \in X_n, \forall s \in \mathbf{Q} \cap [0,t]) = \lim_{i \rightarrow \infty} \mathbf{P}_a^n(\omega : \omega(s) \in X_n, \forall s \in F_i) = 1
	\end{equation*}
	By right-continuity of the paths, the result follows.
\end{proof}

\section{\large \bf  Feynman-Kac at the Finite Level}
\subsection{The Feynman-Kac formula}
For the rest of this article we will often encounter the expression $e^{-tH}$. This is well-defined on account of the self-adjointness of $H$ (see \cite{BD15} for details).

The Feynman-Kac formula for the Hamiltonian $H$ over $K$ says
\begin{align}
	(e^{-tH}f)(x)=\int_K K_t(x,y)f(y)\,dy,\quad f\in L^2(K)\,,
\end{align}
where
\begin{align}
	K_t(x,y) &= \int_{D[0,t]} e^{-\int_0^tv(\omega(s))\, ds} \, d\mathbf{P}_{x,y,t}(\omega) \cdot p_t(y-x)\,.
\end{align}
For a proof in the real case, see, e.g., \cite[Theorem 1, Ch.\ 10]{Var94}. The same proof works over a local field.\\
We now prove that we have a Feynman-Kac formula also at the finite level.
\begin{theorem}[Feynman-Kac at the finite level]\label{f-k}
	\begin{align}\label{f-keq}\begin{split}
			(e^{-tH_n}f)(x)&=\int_{X_n} K^n_t(x,y)f(y)\,d\mu_n(y)\\
			&=q^{-n}\sum_{y\in X_n}K^n_t(x,y)f(y),\quad f\in L^2(X_n)
		\end{split}
	\end{align}
	where
	\begin{multline}\label{ktn}
		K_t^n(x,y)=\int_{D[0,t]} e^{-\int_0^tv_n(\omega(s))\, ds} \, d\mathbf{P}_{x,y,t}^n(\omega) \cdot p_{t,n}(y-x),\quad x,y\in X_n\,.
	\end{multline}
\end{theorem}
\begin{proof}
	By \eqref{etpanf} we have
	\[
	e^{-t\pna} f=p_{t,n}*f\,,
	\]
	where $*$ is convolution on $X_n$:
	\begin{equation*}
		(f \ast g )(x) = \int_{X_n} f(y)g(x-y) \, d\mu_n(y) = q^{-n}\sum_{y \in X_n} f(y)g(x-y)\,.
	\end{equation*}
	This gives
	\begin{equation*}
		(e^{-t\pna/N}e^{-tV_n/N}f)(x) = \int_{X_n} p_{t/N,n}(y-x) e^{-t v_n(y)/N}f(y) \, d\mu_n(y)
	\end{equation*}
	and thus
	\begin{multline*}
		((e^{-t\pna/N}e^{-tV_n/N})^Nf)(x)\\
		= \int_{X_n^N} p_{t/N,n}(x_1-x)\cdots p_{t/N,n}(x_N-x_{N-1})\,\cdot\\\cdot e^{-t (v_n(x_1)+ \cdots + v_n(x_N))/N}f(x_N) \, d\mu_n(x_N)\cdots d\mu_n(x_1)\\
		= q^{-nN}\sum_{X_n^{N}} p_{t/N,n}(x_1-x)\cdots p_{t/N,n}(x_N-x_{N-1}) e^{-t (v_n(x_1)+ \cdots + v_n(x_N))/N}f(x_N)
	\end{multline*}
	Let $t_r = rt/N$ for $1 \leq r \leq N$. Defining $Y_{t_1,\dots,t_N}(\omega)=(\omega(t_1),\dots,\omega(t_N))$ we have
	\begin{align*}
		& \int_{D[0,t]} e^{-(t/N) \sum_{r=1}^N v_n(\omega(rt/N))} f(\omega(t)) \, d\mathbf{P}_{x}^n(\omega)\\ &= \int_{K^N} e^{-t( v_n(x_1) + \cdots + v_n(x_N))/N} f(x_N) \, d\mathbf{P}_{x}^n \circ Y_{t_1,...,t_N}^{-1}(x_1,\dots,x_N )\\
		&= q^{-nN}\sum_{X_n^{N}} p_{t/N,n}(x_1-x)\cdots p_{t/N,n}(x_N-x_{N-1}) e^{-t (v_n(x_1)+ \cdots + v_n(x_N))/N}f(x_N).
	\end{align*}
	Combining these equations we get
	\begin{equation*}
		((e^{-t\pna/N}e^{-tV_n/N})^Nf)(x) = \int_{D[0,t]} e^{-(t/N) \sum_{r=1}^N v_n(\omega(rt/N))} f(\omega(t)) \, d\mathbf{P}_{x}^n(\omega)\,.
	\end{equation*}
	Now let $N\to\infty$. By Trotter's product formula, the left hand side converges to $(e^{-tH_n}f)(x)$. The right hand side converges to $\int_{D[0,t]} e^{-\int_0^t v_n(\omega(s))\, ds} f(\omega(t)) \, d\mathbf{P}_{x}^n(\omega)$ by bounded convergence and by Riemann integrability of $v_n \circ \omega$ over $[0,t]$. This gives, for $f\in L^2(X_n)$,
	\begin{align*}
		(e^{-tH_n}f)(x) &= \int_{D[0,t]} e^{-\int_0^t v_n(\omega(s))\, ds} f(\omega(t)) \, d\mathbf{P}_{x}^n(\omega)\\
		&=\int_{X_n}\left(\int_{D[0,t]}e^{-\int_0^t v_n(\omega(s))\, ds} f(\omega(t)) \, d\mathbf{P}_{x,y,t}^n(\omega)\ptn(y-x)\right)d\mu_n(y)\\
		&=\int_{X_n}\left(\int_{D[0,t]}e^{-\int_0^t v_n(\omega(s))\, ds} \, d\mathbf{P}_{x,y,t}^n(\omega)\ptn(y-x)\right)f(y)\,d\mu_n(y)\\
		&=q^{-n}\sum_{y\in X_n}\left(\int_{D[0,t]}e^{-\int_0^t v_n(\omega(s))\, ds} \, d\mathbf{P}_{x,y,t}^n(\omega)\ptn(y-x)\right)f(y)\\
		&=q^{-n}\sum_{y\in X_n}K^n_t(x,y)f(y)\\
		&=\int_{X_n} K^n_t(x,y)f(y)\,d\mu_n(y)\,.
	\end{align*}
\end{proof}

\subsection*{Elementary properties of the propagator}
We collect here some elementary properties of the propagator $\knt$.

Let $\{e_x\}_{x\in X_n}$ be the canonical basis for $L^2(X_n)$, i.e., $e_x=q^{n/2}\one_x$. Then $\operatorname{Tr}(e^{-tH_n})=\sum_{x\in X_n}\langle e^{-tH_n}e_x,e_x\rangle$ by definition. From the Feynman-Kac formula:
\[\label{kntex}(e^{-tH_n}e_x)(y)=q^{-n}\sum_{z\in X_n}K^n_t(y,z)e_x(z)=q^{-n}K^n_t(y,x)q^{n/2}=q^{-n/2}K^n_t(y,x)\,,
\]
and so
\begin{align}\label{kntexey}
	\begin{split} 
		\langle e^{-tH_n}e_x,e_y\rangle&=\int_{X_n}(e^{-tH_n}e_x)(z)\overline{e_y(z)}\,d\mu_n(z)=\sum_{z\in X_n}(e^{-tH_n}e_x)(z)e_y(z)q^{-n}\\
		&=(e^{-tH_n}e_x)(y)q^{n/2}q^{-n}=q^{n/2}q^{-n}q^{-n/2}K^n_t(y,x)\\&=q^{-n}K^n_t(y,x)
	\end{split}
	\intertext{which gives}\label{knttr}
	\operatorname{Tr}(e^{-tH_n})&=\sum_{x\in X_n}\langle e^{-tH_n}e_x,e_x\rangle=q^{-n}\sum_{x\in X_n}K^n_t(x,x)\,.
\end{align}
Since $\langle \ethn e_x,e_y\rangle=\langle e_x,\ethn e_y\rangle=\overline{\langle \ethn e_y,e_x\rangle}$, we get
\beq\label{conjsymm}
\knt(x,y)=\overline{\knt(y,x)}.
\eeq 
Like any kernel of a semi-group, $\knt$ also satisfies
\beq\label{chainrule} K^n_{t_1+t_2}(x,y)=\int_{X_n}K^n_{t_1}(x,z)K^n_{t_2}(z,y)\,d\mu_n(z)=\sum_{z\in X_n}\qmn K^n_{t_1}(x,z)K^n_{t_2}(z,y)\,.
\eeq 
This follows from the following calculation (for any $f\in L^2(X_n)$):
\begin{multline*}
	(e^{-(t_1+t_2)H_n}f)(x)=(e^{-t_1H_n}e^{-t_2H_n}f)(x)=\sum_{z\in X_n}\qmn K^n_{t_1}(x,z)(e^{-t_2H_n}f)(z)\\
	=\sum_{z\in X_n}\qmn K^n_{t_1}(x,z)\left(\sum_{y\in X_n}\qmn K^n_{t_2}(z,y)f(y)\right)\\
	=\sum_{y\in X_n}\qmn \left(\sum_{z\in X_n}\qmn K^n_{t_1}(x,z)K^n_{t_2}(z,y)\right)f(y)=
	\sum_{y\in X_n}\qmn K^n_{t_1+t_2}(x,y)f(y)\,.
\end{multline*}
Since the last equality holds for all $f\in L^2(X_n)$, \eqref{chainrule} follows.

\section{\large \bf  Finite Approximations}

\subsection{Convergence of Traces}
In this subsection we will show -- under a mild growth condition on the potential -- that the operator $e^{-tH}$ is of trace class and that
\begin{equation}\label{convtraces}
\operatorname{Tr}(e^{-tH_n}) \rightarrow \operatorname{Tr}(e^{-tH})\text{ as $n\to\infty$}.
\end{equation}	
We have
\begin{align*}\operatorname{Tr}(e^{-tH_n})&=\sum_{x\in X_n}\langle e^{-tH_n}e_x,e_x\rangle=q^{-n}\sum_{x\in X_n}K^n_t(x,x)\text{ (from the previous section)}\\
	\label{Mercer}
	\operatorname{Tr}(e^{-tH})&=\int_{K}K_t(x,x)\,dx\text{ (Mercer's Theorem)}\,.
\end{align*} 
The convergence \eqref{convtraces} will be established by proving a suitable convergence at the level of propagators. The proof is patterned on the corresponding proof in \cite{DVV94}[Sec.\ 4].
\begin{lemma}
	\label{expUniformBound}
	There is a constant $B_t$, independent of $n$, such that
	\begin{equation*}
		\sup_{x\in X_n} K^n_t(x,x) \leq B_t\,.
	\end{equation*}
\end{lemma}

\begin{proof}
	This follows from the Feynman-Kac formula and lemma \ref{ptunifbdd}.
\end{proof}

\begin{lemma}
	\label{uniformConvergenceDiag}$K^n_t$ converges continuously to $K_t$, i.e.,
	if $x_n\in X_n \rightarrow x\in K$ and $y_n\in X_n \rightarrow y\in K$ as $n \rightarrow \infty$, then
	\begin{equation*}
		K_t^n(x_n,y_n) \rightarrow K_t(x,y).
	\end{equation*}
	In particular, $K^n_t$ converges uniformly to $K_t$ on compact sets.
\end{lemma}

\begin{proof}
	From Lemma~\ref{uniformConvDensity}	we have that $p_{t,n}(y-x)$ converges uniformly on compacta to $p_t(y-x)$. Uniform convergence on a compact set implies continuous convergence on that set, hence $p_{t,n}(y_n-x_n)\to p_t(y-x)$. From Theorem~\ref{uniformConvergenceMeasures} we have that
	\begin{equation*}
		\int_{D[0,t]} e^{-\int_0^tv_n(\omega(s))\, ds} \, d\mathbf{P}_{x_n,y_n,t}^n(\omega) \rightarrow  \int_{D[0,t]} e^{-\int_0^tv(\omega(s))\, ds} \, d\mathbf{P}_{x,y,t}(\omega)\,.
	\end{equation*}
	The lemma now follows from the Feynman-Kac formula (Thm~\ \ref{f-k}).
\end{proof}

\begin{corollary}
	For any ball $B_m$ we have
	\begin{equation*}
		\sum_{x \in B_m\cap X_n}q^{-n} K^n_t(x,x) \rightarrow \int_{B_m} K_t(x,x) \, dx
	\end{equation*}
	when $n\to\infty.$
\end{corollary}
\begin{proof}
	\begin{align*}
		\sum_{x \in B_m\cap X_n}q^{-n} K^n_t(x,x)=\int_{B_m} K^n_t(x,x)\,d\mu(x)\to \int_{B_m} K_t(x,x)\,d\mu(x)\,,
	\end{align*}	
	where in the second integral the function $x\to K^n_t(x,x)$ is regarded as an element of $\mathcal D_n$.	
\end{proof}	

To prove the convergence of traces \eqref{convtraces} we need to extend the previous result to the whole space $K$. For this we need some results related to L\'{e}vy's inequality. The following lemmas, which are adapted from \cite{Var80} (with $\rr$ replaced by a local field $K$), will do the job. 
\begin{lemma}[L\'{e}vy's Inequality]
	\label{LevyInequ}
	Let $Y_1,...,Y_n$ be independent random variables and let $\epsilon, \delta >0$ and $S_j = Y_1 + \cdots + Y_j$. If
	\begin{equation*}
		{\mathbf P}(|Y_r + \cdots + Y_l| \geq \delta) \leq \epsilon,
	\end{equation*}
	for all $1\leq r\leq l \leq n$, then
	\begin{equation*}
		{\mathbf P}(\sup_{1\leq j \leq n} |S_j| >2 \delta) \leq 2 \epsilon.
	\end{equation*}
\end{lemma}

\begin{proof}
	Define $E = \{\sup_{1 \leq j\leq n} |S_j| \geq 2\delta\}$, $E_1 = \{|S_1| \geq 2\delta \}$, and $E_k = \{|S_1| < 2\delta,...,|S_{k-1}| < 2\delta, |S_k| \geq 2\delta \}$, $k\geq2$. Then
	\begin{equation*}
		{\mathbf P}(E \cap (|S_n| \leq \delta)) = {\mathbf P}(\bigcup_{j=1}^n (E_j \cap (|S_n| \leq \delta))) \leq {\mathbf P}(\bigcup_{j=1}^n (E_j \cap (|S_n-S_j| \geq \delta)))
	\end{equation*}
	By independence
	\begin{equation*}
		{\mathbf P}(\bigcup_{j=1}^n (E_j \cap (|S_n-S_j| \geq \delta)))= \sum_{j=1}^n{\mathbf P}(E_j){\mathbf P}(|S_n-S_j| \geq \delta) \leq \epsilon {\mathbf P}(E)
	\end{equation*}
	Also ${\mathbf P}(E \cap (|S_n| > \delta)) \leq {\mathbf P}(|S_n| > \delta) \leq \epsilon$. Combining the two inequalities gives that
	\begin{equation*}
		{\mathbf P}(E) \leq \frac{\epsilon}{1-\epsilon}.
	\end{equation*}
	If $\epsilon < 1/2$, then ${\mathbf P}(E) \leq 1 < 2\epsilon$. If $\epsilon \geq 1/2,$ then ${\mathbf P}(E) \leq \frac{\epsilon}{1-\epsilon} \leq 2\epsilon$. Thus
	\begin{equation*}
		{\mathbf P}(\sup_{1\leq j \leq n} |S_j| >2 \delta) \leq {\mathbf P}(E) \leq 2\epsilon.
	\end{equation*}
\end{proof}

\begin{lemma}
	Let $Y_t$ be a stochastic process with independent increments. Let $I$ be a finite interval in $[0,\infty)$ and $F$ a finite set of points in $I$. Then for every $k>0$,
	\begin{equation*}
		{\mathbf P}(\sup_{s,t \in F} |Y_t-Y_s| > 4\delta) \leq \frac{2}{\delta^k} \sup_{s,t \in F} E_{\mathbf P}(|Y_t-Y_s|^k).
	\end{equation*}
\end{lemma}

\begin{proof}
	Let $F$ be the $m$ time points $0 \leq t_1 < ... < t_m $. The random variables $Y_1 = Y_{t_2} - Y_{t_1}, ..., Y_{m-1} = Y_{t_m} - Y_{t_{m-1}}$ are independent. Also,
	\begin{equation*}
		|Y_r +...+Y_l| = |Y_{t'} - Y_{t''}|
	\end{equation*}
	for some $t',t'' \in F$.
	Define $\epsilon$ by
	\begin{equation*}
		\epsilon = \sup_{1\leq r \leq l \leq m-1} {\mathbf P}(|Y_r +...+Y_l| \geq \delta ).
	\end{equation*}
	By Chebyshev's inequality
	\begin{equation*}
		{\mathbf P}(|Y_r +...+Y_l| \geq \delta) = {\mathbf P}(|Y_{t'} - Y_{t''}| \geq \delta) \leq \frac{1}{\delta^k} E_{\mathbf P}(|Y_{t'}-Y_{t''}|^k)\,,
	\end{equation*}
	and so
	\begin{equation*}
		\epsilon \leq \frac{1}{\delta^k} E_{\mathbf P}(|Y_{t'}-Y_{t''}|^k)\,.
	\end{equation*}
	By Lemma \ref{LevyInequ},
	\begin{align*}
		&{\mathbf P}(\sup_{s,t \in F} |Y_t-Y_s| > 4\delta) \\&\leq {\mathbf P}(\sup_{1 \leq i\leq m} |Y_{t_i}-Y_{t_1}| > 2\delta) = {\mathbf P}(\sup_{1 \leq i\leq m-1} |Y_1 + ... + Y_i| > 2\delta) \\
		& \leq 2\epsilon \leq \sup_{s,t \in F} \frac{2}{\delta^k} E_{\mathbf P}(|Y_t-Y_s|^k)
	\end{align*}
	
\end{proof}
Returning to our measures $\mathbf P^n_x$, we have
\begin{lemma}
	\label{LevyInequ2}
	Let $I$ be a finite interval in $[0,\infty)$. Then for every $k>0$,
	\begin{equation*}
		{\mathbf P^n_x}(\sup_{s,t \in I} |Y_t-Y_s| > 4\delta) \leq \frac{2}{\delta^k} \sup_{s,t \in I} E_{\mathbf P^n_x}(|Y_t-Y_s|^k).
	\end{equation*}
\end{lemma}

\begin{proof}
	Let $F_n$ be an increasing sequence of finite subsets of $I$ such that $I \cap \mathbf{Q} = \bigcup F_n$. Then
	\begin{multline*}
		{\mathbf P^n_x}(\sup_{s,t \in I \cap \mathbf{Q}} |Y_t-Y_s| > 4\delta)\\ = \lim_{n \rightarrow \infty} {\mathbf P^n_x}(\sup_{s,t \in F_n} |Y_t-Y_s| > 4\delta) \leq \frac{2}{\delta^k} \sup_{s,t \in I \cap \mathbf{Q}} E_{\mathbf P^n_x}(|Y_t-Y_s|^k).
	\end{multline*}
	By right-continuity of the process,
	\begin{equation*}
		{\mathbf P^n_x}(\sup_{s,t \in I} |Y_t-Y_s| > 4\delta) \leq \sup_{s,t \in I} \frac{2}{\delta^k} E_{\mathbf P^n_x}(|Y_t-Y_s|^k)\,.
	\end{equation*}
\end{proof}

Recall that $\alpha$ is the exponent appearing in the definition of the Hamiltonian: $H=P^\alpha + V$.
\begin{proposition}
	Pick a real number $k$ with $0<k<\alpha$, and let $x\in X_n$. There exists a constant $A_k>0$, independent of $n$ and $x$, such that
	\begin{equation*}
		K^n_t(x,x) \leq A_k\cdot \left(e^{-\frac{t}{2} v^{*}(x)} + \frac{1}{|x|^{k}}\right),
	\end{equation*}
	where $v^{*}(x) = \inf_{|y| = |x|} v(y)$.
\end{proposition}

\begin{proof}
	Define $R_1$ and $R_2$ by
	\begin{equation*}
		R_1 = \{\omega | \omega(0)=x, |\omega(s)| = |x|, \, \forall s \in [0,t/2] \}
	\end{equation*}
	\begin{equation*}
		R_2 = \{\omega | \omega(0)=x, |\omega(s)| \neq |x|\text{ for some }s \in [0,t/2] \}
	\end{equation*}
	and
	\begin{equation*}
		I_i = \int_{R_i} e^{-\int_0^t v_n(\omega(s)) \, ds} \mathbf{1}_x(\omega(t)) \, d\mathbf{P}_{x}^n(\omega).
	\end{equation*}
	for $i=1,2$.
	Then by Feynman-Kac,
	\begin{equation*}
		q^{-n}K^n_t(x,x) = I_1 + I_2
	\end{equation*}
	For $\omega \in R_1$
	\begin{equation*}
		\int_0^t v_n(\omega(s)) \, ds \geq \int_0^{t/2} v_n(\omega(s)) \, ds \geq \frac{t}{2} v_n^{*}(x) \geq \frac{t}{2} v^{*}(x),
	\end{equation*}
	so we get 
	\begin{align*}
		I_1 &= \int_{R_1} e^{-\int_0^t v_n(\omega(s)) \, ds} \mathbf{1}_x(\omega(t)) \, d\mathbf{P}_{x}^n(\omega) \leq e^{-\frac{t}{2} v^{*}(x)} \mathbf{P}_{x}^n(R_1 \cap (\omega(t) = x)) \\
		&\leq e^{-\frac{t}{2} v^{*}(x)} \mathbf{P}_{x}^n(R_1) q^{-n} \sup_{y\in X_n} p_{t/2,n}(y) \leq A'_k q^{-n} e^{-\frac{t}{2} v^{*}(x)}.
	\end{align*}
	where  $A'_k$ is independent of $x$ and $n$, and where the next to last inequality follows from a calculation similar to that of \eqref{pt4nineq}.
	Also, by Lemma~\ref{LevyInequ2},
	\begin{align}\begin{split}
			I_2 &\leq \mathbf{P}_{x}^n(R_2 \cap (\omega(t)=x)) \leq \mathbf{P}_{x}^n(R_2) q^{-n} \sup_{y\in X_n} p_{t/2,n}(y) \\
			&\leq A''_k q^{-n} \mathbf{P}_{x}^n(|\omega(s)-\omega(0)| \geq |x| \text{ for some }s \in [0,t/2]) \\
			&\stackrel{\text{Lemma \ref{LevyInequ2}}}{\leq}  q^{-n}A'''_k \frac{2}{|x|^{k}} \sup_{u,s \in [0,t/2]} E_{\mathbf P_x^n}(|Y_u-Y_s|^k) \\
			& \leq q^{-n}A''''_k \frac{1}{|x|^{k}},
		\end{split}
	\end{align}
	where the last inequality follows from the computations in the proof of Proposition~\ref{Centsovprop}.
	Setting $A_k=\max(A'_k,A''''_k)$, this gives
	\begin{equation*}
		K^n_t(x,x) \leq A_k\cdot \left(e^{-\frac{t}{2} v^{*}(x)} + \frac{1}{|x|^{k}} \right)\,.
	\end{equation*}
\end{proof}

\begin{theorem}
	\label{ConvergenceTraces}
	Assume $\alpha>1$. If $\frac{|v(x)|}{\ln(|x|)} \rightarrow \infty$  as $|x| \rightarrow \infty$, then $e^{-tH}$ is of trace class and 
	\begin{equation*}
		\operatorname{Tr}(e^{-tH_n}) \rightarrow \operatorname{Tr}(e^{-tH})
	\end{equation*}
	as $n \rightarrow \infty$.
\end{theorem}

\begin{proof}
	Choose a $k$ with $1<k<\alpha$. 
	By the previous lemma,
	\begin{align}
		\sum_{|x|=q^m} K^n_t(x,x) &\leq (q-1)q^{m-1+n}\cdot A_k\cdot\left(e^{-\frac{t}{2} v^{*}(\beta^{-m})} + \frac{1}{q^{mk}} \right)\\
		&= q^n\cdot A_k\cdot(1-1/q)\left(q^m e^{-\frac{t}{2} v^{*}(\beta^{-m})} + \frac{q^m}{q^{mk}} \right).
	\end{align}
	So we have
	\begin{equation*}
		\sum_{|x|\geq q^m}q^{-n}\cdot K^n_t(x,x) \leq A_k\cdot(1-1/q) \sum_{i\geq m} \left(q^i e^{-\frac{t}{2} v^{*}(\beta^{-i})} + \frac{q^i}{q^{ik}} \right)
	\end{equation*}
	which for $k>1$ goes to $0$ as $m\rightarrow \infty$, uniformly  in $n$.
	Thus
	\begin{equation*}
		\sum_{|x|\geq q^m} q^{-n}\cdot K^n_t(x,x) \rightarrow 0
	\end{equation*}
	as $m \rightarrow \infty$, uniformly in $n$. Since
	\begin{equation*}
		\int_{|x| \geq q^m} K_t(x,x) \, dx = \int_{|x| \geq q^m} \int_{D[0,t]} e^{-\int_0^tv(\omega(s))\, ds} \, d\mathbf{P}_{x,x,t}(\omega) \cdot p_t(0) \, dx\,,
	\end{equation*}
	the same calculations as above show that the integral converges. All of this now shows that
	\begin{equation*}
		\operatorname{Tr}(e^{-tH_n}) = \sum_{x\in X_n} q^{-n}K^n_t(x,x)\to \int_{K} K_t(x,x) \, dx\,.
	\end{equation*}
	By Mercer's Theorem we have $\int_{K} K_t(x,x) \, dx=\operatorname{Tr}(e^{-tH})$, and so
	\begin{equation*}
		\operatorname{Tr}(e^{-tH_n}) \rightarrow \operatorname{Tr}(e^{-tH})\,.
	\end{equation*}
	
\end{proof}

\subsection{Convergence of Eigenvalues and Eigenfunctions}
We first wish to use the fact that
\begin{equation*}
	\operatorname{Tr}(e^{-tH_n}) \rightarrow \operatorname{Tr}(e^{-tH})
\end{equation*}
to prove that $e^{-tH_n}$ converges to $e^{-tH}$ in the trace norm.

From \cite{BD15} we know that $e^{-tH_n}\to e^{-tH}$ strongly. This immediately implies 
\begin{lemma}
	\label{ConvergenceTracesFinite}
	For any operator $L$ of finite rank we have
	\begin{equation*}
		\operatorname{Tr}(e^{-tH_n}L) \rightarrow \operatorname{Tr}(e^{-tH}L)
	\end{equation*}
	as $n \rightarrow \infty$.
\end{lemma}

Let $\mathcal{H}_2$ denote the Hilbert-Schmidt operators with inner product $\langle S,T \rangle = \operatorname{Tr}(T^{*}S)$ and corresponding norm $||\cdot||_2$. Also let $||T||_1 = \operatorname{Tr}(|T|)$ denote the trace norm. 

The proofs of the remaining results of this section follow the same pattern as in \cite{DVV94}, but we include them here for completeness.
\begin{theorem}
	For any $t>0$,
	\begin{equation*}
		||e^{-tH_n}-e^{-tH}||_1 \rightarrow 0
	\end{equation*}
	as $n \rightarrow \infty$.
\end{theorem}

\begin{proof}
	We will first prove that
	\begin{equation*}
		||e^{-tH_n}-e^{-tH}||_2 \rightarrow 0
	\end{equation*}
	as $n \rightarrow \infty$.
	This follows if
	\begin{equation*}
		||e^{-tH_n}||_2 \rightarrow ||e^{-tH}||_2
	\end{equation*}
	as $n \rightarrow \infty$, and
	\begin{equation*}
		\langle e^{-tH_n}, L \rangle \rightarrow \langle e^{-tH}, L \rangle
	\end{equation*}
	for all $L \in \mathcal{H}_2$ as $n \rightarrow \infty$.
	From Proposition \ref{ConvergenceTraces}, we get that
	\begin{equation*}
		||e^{-tH_n}||_2^2 = \operatorname{Tr}(e^{-2tH_n}) \rightarrow \operatorname{Tr}(e^{-2tH}) = ||e^{-tH}||_2^2
	\end{equation*}
	as $n \rightarrow \infty$.
	
	By Lemma \ref{ConvergenceTracesFinite},
	\begin{equation*}
		\langle e^{-tH_n}, L \rangle = \operatorname{Tr}(L^{*}e^{-tH_n}) \rightarrow \operatorname{Tr}(L^{*}e^{-tH}) = \langle e^{-tH}, L \rangle
	\end{equation*}
	for all operators $L$ of finite rank. By density of finite rank operators in $\mathcal{H}_2$, the result follows.
	
	To go from convergence in Hilbert-Schmidt norm to convergence in trace norm, we use the inequality
	\begin{equation*}
		||A^2-B^2||_1 \leq ||(A+B)||_2\cdot||(A-B)||_2 + 2||B||_2\cdot||(A-B)||_2
	\end{equation*}
	which follows from
	\begin{equation*}
		A^2-B^2 = (A+B)(A-B) + (A-B)B - B(A-B).
	\end{equation*}
	With $A=e^{-\frac{t}{2}H_n}$ and $B=e^{-\frac{t}{2}H}$, we get that
	\begin{multline*}
		||e^{-tH_n}-e^{-tH}||_1 \leq ||(e^{-\frac{t}{2}H_n}+e^{-\frac{t}{2}H})||_2
		\cdot||(e^{-\frac{t}{2}H_n}-e^{-\frac{t}{2}H})||_2\\ + 2||e^{-\frac{t}{2}H}||_2\cdot||(e^{-\frac{t}{2}H_n}-e^{-\frac{t}{2}H})||_2
	\end{multline*}
	which goes to $0$ as $n \rightarrow \infty$. This proves the theorem.
\end{proof}
Convergence in trace norm implies convergence in operator norm which gives convergence of eigenvalues and eigenfunctions (see pp. 289-290 in \cite{RS80}). Thus we have proved by stochastic methods the following result, which was the main theorem both in \cite{DVV94} and \cite{BD15} ($\sigma(\cdot)$ denotes spectrum, $\operatorname{r}(\cdot)$ denotes range projection, and $P^A$ denotes spectral measure of an operator $A$):

\begin{theorem}[Main Theorem]\label{mainthm}
	\begin{enumerate}
		\item If $J$ is a compact subset of $[0,\infty)$ with $J\cap\sigma(H)=\emptyset$, then $J\cap\sigma(H_n)=\emptyset$ for large $n$.
		\item If $\gl\in\sigma(H)$, there exists a sequence $(\gl_n)$ with $\gl_n\in\sigma(H_n)$ such that $\gl_n\to\gl$. Further, if $J$ is a compact neighborhood of an eigenvalue $\gl\in\sigma(H)$, not containing any other eigenvalues of $H$, then any sequence $\gl_n$ with $\gl_n\in\sigma(H_n)\cap J$ converges to $\gl$.
		\item Let $\gl$ and $J$ be as in (2). Then $\dim P^{H_n}(J) = \dim P^H(J)$ for large $n$, 
		and for each orthonormal basis $\{e_1,\dots,e_m\}$ for $\operatorname{r}\left(P^{H}(J)\right)$ there is, for each $n$, an orthonormal basis $\{e_1^n,\dots,e^n_m\}$ for $\operatorname{r}\left(P^{H_n}(J)\right)$ such that $\lim_{n\to\infty}e^n_i=e_i$, $i=1,\dots,m$.  
	\end{enumerate}
\end{theorem}
Finally, we are now ready to reap the benefits from using stochastic methods by showing that the eigenfunctions can be chosen to be continuous, and that they converge uniformly on compact sets.
\begin{lemma}
	\label{BoundExp}
	For each $t > 0$, there exists a constant $C=C(t)$ such that for any $h \in L^2(K)$ and any $n$,
	\begin{equation*}
		||e^{-tH_n}h||_{\infty} \leq C||h||_{L^2(K)}, \qquad ||e^{-tH}h||_{\infty} \leq C||h||_{L^2(K)}.
	\end{equation*}
\end{lemma}
\begin{proof}
	First note that $e^{-tH_n}f$ and $e^{-tH}f$ are continuous functions.\\
	By Feynman-Kac,
	\begin{equation*}
		0 \leq K_t(x,y) \leq p_t(y-x).
	\end{equation*}
	By \cite[Lemma~2, Sec.~4]{Var97} we know that $p_t$ is in $L^2(K)$. Thus for every $x \in K$,
	\begin{equation*}
		|e^{-tH}h(x)| = |\int_{K}K_t(x,y)h(y) \, dy| \leq \int_{K}p_t(y-x)|h(y)| \, dy \leq ||p_t||_{L^2(K)}\cdot ||h||_{L^2(K)}
	\end{equation*}
	Interpreting $h$ as a function on $X_n$, we have by the finite Feynman-Kac formula \eqref{f-keq}:
	\beqstar
	(e^{-tH_n}h)(x) =\sum_{y\in X_n}q^{-n}K^n_t(x,y) \cdot h(y),
	\eeqstar
	so ($B_t=\text{the constant from Lemma~\ref{expUniformBound}}$)
	\begin{multline*}
		|(e^{-tH_n}h)(x)|^2=|\sum_{y\in X_n}\qmn\knt(x,y)h(y)|^2\\
		\leq (\sum_{y\in X_n}q^{-n}|\knt(x,y)|^2)\cdot(\sum_{y\in X_n}\qmn|h(y)|^2)\\
		\stackrel{\eqref{conjsymm}}{=}(\sum_{y\in X_n}q^{-n}\knt(x,y)\knt(y,x))\cdot||h||^2_{L^2(X_n)}\stackrel{\eqref{chainrule}}{=}K^n_{2t}(x,x)\cdot||h||^2_{L^2(X_n)}\\
		\stackrel{\text{Lemma } \ref{expUniformBound}}{\leq} B_{2t}\cdot||h||^2_{L^2(K)}\,.
	\end{multline*}
	So with $C=\max(||p_t||_{L^2(K)},B_{2t})$, the lemma follows.
\end{proof}

\begin{lemma}
	\label{UniformCompactaExp}
	Fix $t>0$. Then for each $h \in L^2(K)$,
	\begin{equation*}
		e^{-tH_n}h \rightarrow e^{-tH}h
	\end{equation*}
	uniformly on compacta.
\end{lemma}

\begin{proof}
	We will prove it for a Schwartz-Bruhat function $h$, and then the general result follows from Lemma \ref{BoundExp}. Let $J$ be the union of a finite set of balls which cover the support of $h$. We have $K_t^n \rightarrow K_t$ uniformly on compacta in $K \times K$ (Lemma~\ref{uniformConvergenceDiag}). As $h$ is Schwartz-Bruhat, $D_nh = h$ for $n$ sufficiently large. Thus $K_t^nh \rightarrow K_th$ uniformly on compacta. Let $L$ be a compact set. Then for $x \in L$
	\begin{equation*}
		|e^{-tH_n}h(x)- e^{-tH}h(x)| \leq \int_J |K_t^n(x,y)h(y) - K_t(x,y)h(y)| \, dy  \rightarrow 0 
	\end{equation*}
	as $n\to \infty$, uniformly in $x$.
\end{proof}

\begin{theorem}[Uniform Convergence on Compacta for Eigenfunctions]
	Let $f_{n,j}$ and $f_j$ be eigenfunctions of $H_n$ and $H$ corresponding to the eigenvalues $\lambda_{n,j}$ and $\lambda_j$ respectively. Assume that $\lambda_{n,j}$ converges to $\lambda_j$ and that $f_{n,j}$ converges to $f_j$ in $L^2(K)$. Then
	\begin{equation*}
		f_{n,j} \rightarrow f_j\quad\text{as $n\to\infty$}
	\end{equation*}
	uniformly on compacta.
\end{theorem}

\begin{proof}
	We will first prove that
	\begin{equation*}
		e^{-tH_n}f_{n,j} \rightarrow e^{-tH}f_j\quad\text{as $n\to\infty$}
	\end{equation*}
	uniformly on compacta.
	Let $M$ be a compact set. We have
	\begin{multline*}
		||e^{-tH_n}f_{n,j} -e^{-tH}f_j||_{L^{\infty}(M)}\\ \leq ||e^{-tH_n}f_{n,j} -e^{-tH_n}f_j||_{L^{\infty}(M)} + ||e^{-tH_n}f_j -e^{-tH}f_j||_{L^{\infty}(M)}
	\end{multline*}
	This goes to $0$ by Lemma \ref{BoundExp} and  \ref{UniformCompactaExp}.
	
	Now we know that, as $n\to\infty$,
	\begin{equation*}
		e^{-t\lambda_{n,j}}f_{n,j} = e^{-tH_n}f_{n,j} \rightarrow e^{-tH}f_j = e^{-t\lambda_j}f_j
	\end{equation*}
	uniformly on compacta. Since $e^{-t\lambda_{n,j}} \rightarrow e^{-t\lambda_j}$, the result follows.
\end{proof}

\section*{Acknowledgment}
The two first named authors would like to thank the UCLA Math Department -- and Professor Varadarajan in particular -- for the hospitality afforded to them during their stay in Fall 2014/Winter 2015.

\providecommand{\bysame}{\leavevmode\hbox to3em{\hrulefill}\thinspace}
\providecommand{\MR}{\relax\ifhmode\unskip\space\fi MR }
\providecommand{\MRhref}[2]{%
	\href{http://www.ams.org/mathscinet-getitem?mr=#1}{#2}
}
\providecommand{\href}[2]{#2}

\end{document}